\documentclass{article}

\pdfoutput=1

\usepackage{tikz}
\usepackage{blkarray}
\usetikzlibrary{calc}
\usepackage{hyperref}

\input{vsynth.sty}

\title{Optimal ancilla-free Clifford+$V$ approximation of
  $z$-rotations}

\author{\begin{tabular}{c}
    Neil J. Ross \\[.5ex]
    \normalsize Department of Mathematics and Statistics \\
    \normalsize Dalhousie University
  \end{tabular}
}

\date{}

\begin{document}

\maketitle

\begin{abstract}
  We describe a new efficient algorithm to approximate $z$-rotations
  by ancilla-free Clifford+$V$ circuits, up to a given precision
  $\epsilon$. Our algorithm is optimal in the presence of an oracle
  for integer factoring: it outputs the shortest Clifford+$V$ circuit
  solving the given problem instance. In the absence of such an
  oracle, our algorithm is still near-optimal, producing circuits of
  $V$\!-count $m + O(\log(\log(1/\epsilon)))$, where $m$ is the
  $V$\!-count of the third-to-optimal solution. A restricted version
  of the algorithm approximates $z$-rotations in the Pauli+$V$ gate
  set. Our method is based on previous work by the author and Selinger
  on the optimal ancilla-free approximation of $z$-rotations using
  Clifford+$T$ gates and on previous work by Bocharov, Gurevich, and
  Svore on the asymptotically optimal ancilla-free approximation of
  $z$-rotations using Clifford+$V$ gates.
\end{abstract}

\section{Introduction}
\label{sec-intro}

\subsection{The synthesis problems}
\label{subsec-synth-prob}

The \emph{unitary group of order 2}, denoted $U(2)$, is the group of
$2\by 2$ complex unitary matrices. We also refer to the elements of
this group as operators, or \emph{gates}. The \emph{special unitary
  group of order 2}, denoted by $SU(2)$, is the subset of $U(2)$
consisting of unitary matrices of determinant 1. We will be concerned
with the notion of distance that arises from the operator norm, that
is, for $U$ and $U'$ in $U(2)$:
\[
\norm{U-U'} = \mbox{sup}\s{|Uv - U'v| \such |v|=1}.
\]
We refer to subsets of $U(2)$ as \emph{gate bases} and to a finite
word $W$ over a gate base $B$ as a \emph{circuit over $B$}. By a
slight abuse of notation, we write $W$ to denote both a circuit over
$B$ and the unitary obtained by multiplying the basis elements
composing $W$.

We are interested in decomposing, or \emph{synthesizing}, unitary
matrices into circuits over a given gate base. For a gate base $B$ and
unitary matrix $U$, the decomposition of $U$ over $B$ can be done
\emph{exactly}, if there exists a circuit $W$ over $B$ such that
$W=U$, or \emph{approximately up to some $\epsilon>0$}, if there
exists a circuit $W$ over $B$ such that $\norm{U-W}\leq \epsilon$. We
thus get the following two problems.

\begin{itemize}
\item \emph{Exact synthesis problem for $B$:} given a unitary $U$,
  determine whether there exists a circuit $W$ over $B$ such that
  $W=U$ and, in case such a circuit exists, construct one.
\item \emph{Approximate synthesis problem for $B$:} given a unitary
  $U$ and a precision $\epsilon\geq 0$, determine whether there exists
  a circuit $W$ over $B$ such that $\norm{W-U}\leq\epsilon$ and, in
  case such a circuit exists, construct one.
\end{itemize}

In what follows, we focus on finite gate bases. If $B$ is such a gate
base, then the set of circuits over $B$ is countable. Since $U(2)$ is
uncountable, this implies that the exact synthesis problem for $B$
will sometimes be solved negatively: there are unitary matrices that
cannot be exactly synthesized over $B$. However, if the set of
circuits over $B$ is dense in $U(2)$, then the approximate synthesis
problem for $B$ can always be solved positively.

Because the state of a qubit is defined up to scaling by a unit
scalar, the synthesis of a unitary $U$ is sometimes done \emph{up to a
  phase}.  This means that instead of finding a circuit $W$ such that
$\norm{U-W}\leq \epsilon$, one looks for a circuit $W$ and a unit
scalar $\lambda$ such that $\norm{U-\lambda W} \leq \epsilon$. This
defines a third synthesis problem.

\begin{itemize}
\item \emph{Approximate synthesis problem for $B$ up to a phase:}
  given a unitary $U$ and a precision $\epsilon\geq 0$, determine
  whether there exists a circuit $W$ over $B$ and a unit scalar
  $\lambda$ such that $\norm{U-\lambda W}\leq\epsilon$ and, in case
  such a circuit exists, construct one.
\end{itemize}

Since a global phase has no observable effect in quantum mechanics, it
is often sufficient to define a decomposition method for special
unitary matrices. Indeed, suppose that $B$ is a gate base such that
the set of circuits over $B$ is dense in $SU(2)$. If we have an
algorithm to approximately synthesize elements of $SU(2)$ into
circuits over $B$, then we can synthesize arbitrary unitary matrices
over $B$ up to a phase, since the determinant of a unitary matrix
always has norm 1.

A decomposition method solving any of the above three problems is
evaluated with respect to its \emph{time complexity} (what is its
run-time?) and to its \emph{circuit complexity} (how many gates are
contained in the produced circuit?).

\subsection{Synthesis of $z$-rotations using $V$\!-gates}
\label{subsec-z-rot-V}

We are interested in the following \emph{$V$\!-gates}
\[
V_X = \frac{1}{\sqrt 5}(I+2iX) = \frac{1}{\sqrt 5}
\begin{pmatrix}
  1 & 2i \\
  2i & 1
\end{pmatrix}, \quad V_Y = \frac{1}{\sqrt 5}(I+2iY) = \frac{1}{\sqrt
  5}
\begin{pmatrix}
  1 & 2 \\
  -2 & 1
\end{pmatrix}, \mbox{ and }
\]
\[
V_Z = \frac{1}{\sqrt 5}(I+2iZ) = \frac{1}{\sqrt 5}
\begin{pmatrix}
  1 + 2i & 0 \\
  0 & 1-2i
\end{pmatrix},
\]
and their adjoints
\[
V_X\da = \frac{1}{\sqrt 5}(I-2iX) = \frac{1}{\sqrt 5}
\begin{pmatrix}
  1 & -2i \\
  -2i & 1
\end{pmatrix}, \quad V_Y\da = \frac{1}{\sqrt 5}(I-2iY) =
\frac{1}{\sqrt 5}
\begin{pmatrix}
  1 & -2 \\
  2 & 1
\end{pmatrix}, \mbox{ and }
\]
\[
V_Z\da = \frac{1}{\sqrt 5}(I-2iZ) = \frac{1}{\sqrt 5}
\begin{pmatrix}
  1 - 2i & 0 \\
  0 & 1+2i
\end{pmatrix}.
\]
It was shown in \cite{LPS1986} and \cite{LPS1987} that the group
generated by the $V$\!-gates is dense in $SU(2)$. It was later shown
in \cite{HRC2002} that for any operator $U\in SU(2)$ and any precision
$\epsilon$, there exists an approximation for $U$ over $V=\s{V_X, V_Y,
  V_Z, V_X\da, V_Y\da, V_Z\da}$ that requires only $O(\log
(1/\epsilon))$ gates. However, no approximate synthesis algorithm was
provided. In \cite{BGS2013}, Bocharov, Gurevich, and Svore defined a
probabilistic algorithm for the approximate synthesis of unitaries
over the Pauli+$V$ gate set, which consists of the $V$\!-gates
together with the Pauli gates $X$, $Y$, and $Z$. Because the Pauli
gates form a subgroup of the Clifford gates, the algorithm of
\cite{BGS2013} is also a synthesis algorithm for the Clifford+$V$ gate
set, which consists of the $V$\!-gates together with the Clifford
gates, whose generators are:
\[
\omega = e^{i\pi/4}, \quad S =
\begin{pmatrix}
  1 & 0 \\
  0 & i
\end{pmatrix}, \quad \mbox{and} \quad H = \frac{1}{\sqrt 2}
\begin{pmatrix}
  1 & 1 \\
  1 & -1
\end{pmatrix}.
\]
In the context of the Clifford+$V$ gate set, the complexity of a
circuit is measured by counting the number of $V$\!-gates appearing in
it, its \emph{$V$\!-count}. This is due to the fact that the Clifford
operators can always be moved to the end of a circuit using equations
such as $\omega V_X= V_X \omega$, $SV_X = V_Y S$, $HV_X=V_Z H$, and so
on.

The algorithm of \cite{BGS2013} is efficient in the sense that it runs
in probabilistic polynomial time. Moreover, it yields circuits of
$V$\!-count bounded above by $12\log_5(2/\epsilon)$ for arbitrary
unitaries.

The method of \cite{BGS2013} was adapted from the one developed in
\cite{Selinger-newsynth} for the Clifford+$T$ gate set. It relies on
the definition of an algorithm for the Clifford+$V$ decomposition of
\emph{$z$-rotations}, i.e., matrices of the form
\[
\Rz (\theta) =
\begin{pmatrix}
  e ^{-i\theta/2} & 0 \\
  0 & e ^{i\theta/2}
\end{pmatrix}.
\]
For these gates, the algorithm of \cite{BGS2013} achieves circuits of
$V$\!-count bounded above by $4\log_5(2/\epsilon)$. Such an algorithm
can then be used for the synthesis of an arbitrary element $U$ of
$SU(2)$ by first writing $U$ as a product of three $z$-rotations using
Euler angles
\[
U=\Rz(\theta_1)X\Rz(\theta_2)X\Rz(\theta_3)
\]
and then applying the algorithm to each of the $\Rz(\theta_i)$.

\subsection{Results}
\label{subsec-contrib}

In the present paper, we define an efficient and optimal algorithm for
the approximate synthesis of $z$-rotations over the Clifford+$V$ gate
set.  Our algorithm is defined by adapting techniques developed in
\cite{gridsynth} for the Clifford+$T$ gate set. We stress that the
algorithm is \emph{literally optimal}, i.e., for any given pair
$(\theta,\epsilon)$ of an angle and a precision, the algorithm finds
the shortest possible ancilla-free Clifford+$V$ circuit $W$ such that
$\norm{W - \Rz(\theta)}\leq\epsilon$. As in \cite{gridsynth}, the
optimality of the algorithm depends on the presence of a factoring
oracle. Because of Shor's algorithm \cite{Shor}, a quantum computer
can serve as such an oracle. For this reason, the algorithm is
actually an efficient and optimal \emph{quantum} synthesis
algorithm. However, the \emph{classical} algorithm obtained in the
absence of a factoring oracle is efficient and nearly optimal: in this
case the algorithm produces circuits of $V$\!-count $m +
O(\log(\log(1/\epsilon)))$, where $m$ is the $V$\!-count of the
third-to-optimal solution. These properties of the classical algorithm
are established under a mild number-theoretic assumption.

We also describe a restricted version of the algorithm which
synthesizes $z$-rotations over the Pauli+$V$ gate set. This restricted
algorithm is also efficient and optimal, if a factoring oracle is
available, and efficient, but only near-optimal, otherwise.

\subsection{Related work}
\label{subsec-related}

Independently of the present paper, in \cite{BBG2014}, Blass,
Bocharov, and Gurevich defined an algorithm for the approximate
synthesis of $z$-rotations in the Pauli+$V$ basis. Their method is in
principle similar to ours, but they use a different technique to solve
the \emph{grid problems} of Section~\ref{ssect-grid-pbs}.

\section{Preliminaries}
\label{sec-prelim}

We write $\N$ for the semiring of non-negative integers, $\Z$ for the
ring of integers and $\C$ for the field of complex numbers. The
conjugate of a complex number is given by $(a+ib)\da =a-ib$. The
Gaussian integers $\Z[i]$ are the complex numbers whose real and
imaginary parts are both integral, i.e., the complex numbers $a+ib$
with $a,b\in\Z$. The units of $\Z[i]$ are $\pm1, \pm i$. Finally, the
group of Pauli operators is generated by the following matrices:
\[
X = \begin{pmatrix}
  0 & 1 \\
  1 & 0
\end{pmatrix}, \quad Y =
\begin{pmatrix}
  0 & -i  \\
  i & 0
\end{pmatrix}, \quad \mbox{and} \quad Z =
\begin{pmatrix}
  1 & 0 \\
  0 & -1
\end{pmatrix}.
\]
The Pauli group is a subgroup of the Clifford group. We write
Pauli+$S$ for the subgroup of the Clifford group generated by the
Pauli gates and the $S$ gate.

\section{Clifford+$V$ Exact Synthesis of Unitaries}
\label{sec-exact-synth}

In this section, we describe an algorithm to solve the problem of
exact synthesis in the Clifford+$V$ gate set. This material is adapted
from \cite{BGS2013}, where an algorithm for exact synthesis in the
Pauli+$V$ gate set was described using the theory of quaternions. We
also use some techniques developed in \cite{Giles-Selinger} for exact
synthesis in the Clifford+$T$ gate set.

\begin{problem}
  \label{pb-exact-synth}
  Given a unitary operator $U\in U(2)$, determine whether there exists
  a Clifford+$V$ circuit $W$ such that $U=W$ and, in case such a
  circuit exists, construct one whose $V$\!-count is minimal.
\end{problem}

To solve Problem~\ref{pb-exact-synth}, we consider unitary matrices of
the form
\begin{equation}
  \label{eq-matrix}
  U = \frac{1}{\rf k} \frac{1}{\rt \ell}\begin{pmatrix} 
    \alpha & \beta \\
    \gamma & \delta
  \end{pmatrix}, \quad \mbox{where $k,\ell\in\N$,
    $\alpha,\beta,\gamma,\delta\in\Z[i],$ and $0\leq \ell\leq 2$.}
\end{equation}
The integers $k$ and $\ell$ in (\ref{eq-matrix}) are called the
\emph{$\sqrt{5}$-denominator exponent} and the
\emph{$\sqrt{2}$-denominator exponent} of $U$ respectively. The least
$k$ (resp. $\ell$) such that $U$ can be written as above is the
\emph{least $\sqrt{5}$-denominator exponent} (resp. \emph{least
  $\sqrt{2}$-denominator exponent}) of $U$. These notions extend
naturally to vectors and scalars of the form
\begin{equation}
  \label{eq-vector}
  \frac{1}{\rf{k}} \frac{1}{\rt{\ell}}\begin{pmatrix} 
    \alpha \\
    \gamma 
  \end{pmatrix}
  \quad                           
  \mbox{ and }
  \quad
  \frac{1}{\rf{k}}\frac{1}{\rt{\ell}}~ \alpha,                            
\end{equation}
where $k,\ell\in\N$, $\alpha, \gamma\in\Z[i]$ and $0\leq \ell\leq
2$. In what follows, we refer to the pair $(k,\ell)$ as the
\emph{denominator exponent} of a matrix, vector, or scalar. It is then
understood that the first component of the pair is the $\sqrt
5$-exponent, while the second is the $\sqrt 2$-exponent. Note that the
least denominator exponent of a matrix, vector, or scalar is the pair
$(k,\ell)$, where $k$ and $\ell$ are the least $\sqrt 5$- and $\sqrt
2$-exponents respectively.

We will show that a unitary operator $U$ can be expressed as a
Clifford+$V$ circuit if and only if it is of the form
(\ref{eq-matrix}) and its determinant is a power of $i$. We start by
showing the left-to-right implication.

\begin{lemma}
  \label{lem-standard-form}
  If $U$ is a Clifford+$V$ operator, then $U=ABC$ where $A$ is a
  product of $V$\!-gates, $B$ is a Pauli+$S$ operator, and $C$ is one
  of $I$, $H$, $HS$, $\omega$, $H\omega$, and $HS\omega$.
\end{lemma}

\begin{proof}
  Clifford gates and $V$\!-gates can be commuted in the sense that for
  every pair $C,V$ of a Clifford gate and a $V$\!-gate, there exists a
  pair $C',V'$ such that $CV=V'C'$. This implies that a Clifford+$V$
  operator $U$ can always be written as $U=AA'$, where $A$ is a
  product of $V$\!-gates and $A'$ is a Clifford operator. Furthermore,
  the Pauli+$S$ group has index 6 as a subgroup of the Clifford group
  and its cosets are: Pauli+$S$, Pauli+$S\cdot H$, Pauli+$S\cdot HS$,
  Pauli+$S\cdot\omega$, Pauli+$S\cdot H\omega$, and Pauli+$S\cdot
  HS\omega$. It thus follows that a Clifford operator $A'$ can always
  be written as $A'=BC$ with $B$ a Pauli+$S$ operator and $C$ one of
  $I$, $H$, $HS$, $\omega$, $H\omega$, and $HS\omega$.
\end{proof}

To show, conversely, that every matrix of the form (\ref{eq-matrix})
whose determinant is a power of $i$ can be represented by a
Clifford+$V$ circuit, we proceed as in \cite{Giles-Selinger}. We show
that every unit vector of the form
(\ref{eq-vector}) can be reduced to $e_1=\left( \begin{smallmatrix} 1 \\
    0\end{smallmatrix}\right)$ by applying a sequence of carefully
chosen Clifford+$V$ gates. Then, we show how applying this method to
the fist column of a unitary matrix $U$ of the form (\ref{eq-matrix})
yields a Clifford+$V$ circuit for $U$.

\begin{lemma}
  \label{lem-l-invariant}
  If $u$ is a unit vector of the form (\ref{eq-vector}) with least
  $\sqrt 5$-denominator exponent $k$ and $W$ is a Clifford circuit,
  then $Wu$ has least $\sqrt 5$-denominator exponent $k$.
\end{lemma}

\begin{proof}
  It suffices to show that the generators of the Clifford group
  preserve the least $\sqrt 5$-denominator exponent of $u$. The
  general result then follows by induction. To this end, write $u$ as
  in (\ref{eq-vector}), with $\alpha=a+ib$ and $\gamma=c+id$:
  \[
  u=\frac{1}{\rf{k}} \frac{1}{\rt{\ell}}\begin{pmatrix}
    a+ib \\
    c+id
  \end{pmatrix}.
  \]
  Now apply $H$, $\omega$, and $S$ to $u$:
  \[
  Hu=\frac{1}{\rf{k}} \frac{1}{\rt{\ell+1}}\begin{pmatrix}
    (a+c)+i(b+d) \\
    (a-c)+i(b-d)
  \end{pmatrix}, \quad \omega u = \frac{1}{\rf{k}}
  \frac{1}{\rt{\ell+1}}\begin{pmatrix}
    (a-b)+i(a+b) \\
    (c-d)+i(c+d)
  \end{pmatrix},
  \]
  \[
  Su = \frac{1}{\rf{k}} \frac{1}{\rt{\ell}}\begin{pmatrix}
    a+ib \\
    -d+ic
  \end{pmatrix}.
  \]
  By minimality of $k$, one of $a,b,c,d$ is not divisible by 5. The
  least $\sqrt 5$-denominator of $Su$ is therefore $k$. Moreover, for
  any two integers $x$ and $y$, $x+y\equiv x-y \equiv 0 ~(\mymod 5)$
  implies $x\equiv y \equiv 0 ~(\mymod 5)$. Thus the least $\sqrt
  5$-denominator exponent of $Hu$ and $\omega u$ is also $k$.
\end{proof}

\begin{lemma}
  \label{lem-column-l}
  If $u$ is a unit vector of the form (\ref{eq-vector}) with least
  denominator exponent $(k,\ell)$, then there exists a Clifford
  circuit $W$ such that $Wu$ has least denominator exponent $(k,0)$.
\end{lemma}

\begin{proof}
  By Lemma~\ref{lem-l-invariant}, we need not worry about $k$ and only
  have to focus on reducing $\ell$. Write $u$ as in (\ref{eq-vector}),
  with $0\leq \ell \leq 2$, $\alpha=a+ib$, and $\gamma=c+id$. Since
  $u$ has unit norm, we have $a^2+b^2+c^2+d^2 = 5^k2^\ell$. We prove
  the lemma by case distinction on $\ell$. If $\ell=0$, there is
  nothing to prove. The remaining cases are treated as follows.
  \begin{itemize}
  \item $\ell=1$. In this case $a^2 + b^2 +c^2 +d^2 = 5^k\cdot 2
    \equiv 0 ~(\mymod 2)$. Therefore only an even number amongst
    $a,b,c,d$ can be odd. Using a Pauli+$S$ operator, we can without
    loss of generality assume that $a\equiv c ~(\mymod 2)$ and
    $b\equiv d ~(\mymod 2)$ or that $a\equiv b ~(\mymod 2)$ and
    $c\equiv d ~(\mymod 2)$. It then follows that either $Hu$ or
    $\omega u$ has denominator exponent $(k,0)$ since
    \[
    Hu = \frac{1}{\rf k}\frac{1}{2}
    \begin{pmatrix}
      (a+c) + i (b+d) \\
      (a-c) + i(b-d)
    \end{pmatrix} \quad \mbox{ and } \quad \omega u = \frac{1}{\rf
      k}\frac{1}{2}
    \begin{pmatrix}
      (a-b) + i (a+b) \\
      (c-d) + i(c+d)
    \end{pmatrix}.
    \]
  \item $\ell=2$. In this case $a^2 + b^2 +c^2 +d^2 = 5^k\cdot 4
    \equiv 0 ~(\mymod 4)$. This implies that $a,b,c$ and $d$ must have
    the same parity and thus, by minimality of $\ell$, must all be
    odd. Using a Pauli+$S$ operator, we can without loss of generality
    assume that $a\equiv b \equiv c \equiv d \equiv 1 ~(\mymod 4)$. It
    then follows that $H\omega u$ has denominator exponent $(k,0)$
    since
    \[
    H\omega u = \frac{1}{\rf k}\frac{1}{4}
    \begin{pmatrix}
      (a-b+c-d) + i (a+b+c+d) \\
      (a-b-c+d) + i(a+b-c-d)
    \end{pmatrix}.
    \]\qedhere
  \end{itemize}
\end{proof}

\begin{remark}
  Let $V$ be one of the $V$\!-gates, $u$ be a vector of the form
  (\ref{eq-vector}), and $k$ and $k'$ be the least $\sqrt
  5$-denominator exponents of $u$ and $Vu$ respectively. Then $k'\leq
  k+1$. Moreover, If it were the case that $k'<k-1$, then the least
  $\sqrt 5$-denominator exponent of $V\da Vu=u$ would be strictly less
  $k$ which is absurd. Thus $k-1 \leq k' \leq k+1$.
\end{remark}

\begin{lemma}
  \label{lem-column-k}
  If $u$ is a unit vector of the form (\ref{eq-vector}) with least
  denominator exponent $(k,0)$, then there exists a Pauli+$V$ circuit
  $W$ of $V$\!-count $k$ such that $Wu = e_1$, the first standard
  basis vector.
\end{lemma}

\begin{proof}
  Write $u$ as in (\ref{eq-vector}) with $\ell=0$, $\alpha=a+ib$, and
  $\gamma=c+id$. Since $u$ has unit norm, we have $a^2+b^2+c^2+d^2 =
  5^k2^0=5^k$. We prove the lemma by induction on $k$.
  \begin{itemize}
  \item $k=0$. In this case $a^2+b^2+c^2+d^2 = 1$. It follows that
    exactly one of $a,b,c,d$ is $\pm 1$ while all the others are
    0. Then $u$ can be reduced to $e_1$ by acting on it using a Pauli
    operator.
  \item $k>0$. In this case $a^2+b^2+c^2+d^2\equiv 0~(\mymod 5)$.  We
    will show that there exists a Pauli+$V$ operator $U$ of
    $V$\!-count 1 such that the least denominator exponent of $Uu$ is
    $k-1$. It then follows by the induction hypothesis that there
    exists $U'$ of $V$\!-count $k-1$ such that $U'Uu=e_1$, which then
    completes the proof.
    
    Consider the residues modulo 5 of $a,b,c,$ and $d$. Since $0,1,$
    and $4$ are the only squares modulo 5, then, up to a reordering of
    the tuple $(a,b,c,d)$, we must have:
    \[
    (a,b,c,d) \equiv \left\{ \begin{array}{l}
        (0,0,0,0) \\
        (\pm2,\pm1,0,0) \\
        (\pm2, \pm2, \pm1, \pm1).
      \end{array} \right. 
    \]
    However, by minimality of $k$, we know that $a\equiv b\equiv
    c\equiv d \equiv 0$ is impossible, so the other two cases are the
    only possible ones. We treat them in turn.
    
    First, assume that one of $a,b,c,d$ is congruent to $\pm 2$, one
    is congruent to $\pm 1$, and the remaining two are congruent to
    $0$. By acting on $u$ with a Pauli operator, we can moreover
    assume without loss of generality that $a\equiv 2$. Now if
    $b\equiv 1$, consider $V_Zu$:
    \[
    V_Zu = \frac{1}{\sqrt{5}^{k+1}}
    \begin{pmatrix}
      (a-2b) + i (2a+b) \\
      (c+2d) + i (d-2c)
    \end{pmatrix}.
    \]
    Since $a\equiv 2$, $b\equiv 1$, and $c\equiv d\equiv 0$, we get
    $(a-2b)\equiv (2a+b)\equiv (c+2d)\equiv (d-2c)\equiv 0~(\mymod
    5)$. The least denominator exponent of $V_Zu$ is therefore
    $k-1$. If on the other hand $b\equiv -1$ then
    \[
    {V_Z}\da u = \frac{1}{\sqrt{5}^{k+1}}
    \begin{pmatrix}
      (a+2b) + i (b-2a) \\
      (c-2d) + i (d+2c)
    \end{pmatrix}
    \]
    and reasoning analogously shows that the least denominator
    exponent of ${V_Z}\da u$ is $k-1$. A similar argument can be made
    in the remaining cases, i.e., when $c\equiv \pm 1$ or $d\equiv
    \pm1$. For brevity, we list the desired operators in the table
    below. The left column describes the residues of $a,b,c$, and $d$
    modulo 5 and the right column gives the operator $U$ such that
    $Uu$ has least denominator exponent $k-1$.
    \begin{center}
      \begin{tabular}{r|l}
        $(a,b,c,d)$   & $U$         \\
        \hline
        \hline \rule{0pt}{3ex} 
        $(2,1,0,0) $  & $V_Z$       \\
        $(2,0,1,0)$   & ${V_Y}\da$  \\
        $(2,0,0,1)$   & $V_X$       \\
        $(2,-1,0,0)$  & $V_Z\da$    \\  
        $(2,0,-1,0)$  & $V_Y$       \\
        $(2,0,0,-1)$  & ${V_X}\da$           
      \end{tabular} 
    \end{center}
    
    Now assume that two of $a,b,c,d$ are congruent to $\pm2$ while the
    remaining two are congruent to $\pm1$. We can use Pauli operators
    to guarantee that $a\equiv 2$ and $c\geq 0$. As above, we list the
    desired operators in a table for conciseness. It can be checked
    that in each case the given operator is such that the least
    denominator exponent of $Uu$ is $k-1$.
    \begin{center}
      \begin{tabular}{r|l}
        $(a,b,c,d)$   & $U$         \\
        \hline
        \hline \rule{0pt}{3ex}
        $(2,2,1,1)$   & ${V_Y}\da$  \\    
        $(2,1,2,1)$   & $V_X$       \\
        $(2,1,1,2)$   & $V_Z$       \\
        $(2,1,2,-1)$  & $V_Z$       \\
        $(2,-1,2,1)$  & ${V_Z}\da$  \\
        $(2,2,1,-1)$  & ${V_X}\da$  \\
        $(2,-2,1,1)$  & $V_X$       \\
        $(2,1,1,-2)$  & ${V_Y}\da$  \\
        $(2,-1,1,2)$  & ${V_Y}\da$  \\
        $(2,-1,1,-2)$ & ${V_Z}\da$  \\
        $(2,-1,2,-1)$ & ${V_X}\da$  \\
        $(2,-2,1,-1)$ & ${V_Y}\da$           
      \end{tabular} 
    \end{center}\qedhere
  \end{itemize}
\end{proof}

We can now solve Problem~\ref{pb-exact-synth}.

\begin{proposition}
  \label{prop-exact-clifford}
  A unitary operator $U\in U(2)$ is exactly representable by a
  Clifford+$V$ circuit if and only if $U$ is of the form
  (\ref{eq-matrix}) and $\det(U)=i^n$ for some integer $n$. Moreover,
  there exists an efficient algorithm that computes a Clifford+$V$
  circuit for $U$ with $V$\!-count equal to the least $\sqrt
  5$-denominator exponent of $U$, which is minimal.
\end{proposition}

\begin{proof}
  The left-to-right implication follows from
  Lemma~\ref{lem-standard-form} and the observation that all the
  generators of the Clifford+$V$ group have determinant $i^n$ for some
  integer $n$. For the right-to-left implication, it suffices to show
  that there exists a Clifford+$V$ circuit $W$ of $V$\!-count $k$ such
  that $WU = I$, since we then have $U=W\da$. To construct $W$, apply
  Lemma~\ref{lem-column-l} and Lemma~\ref{lem-column-k} to the first
  column $u_1$ of $U$. This yields a circuit $W'$ such that the first
  column of $W'U$ is $e_1$. Since $W'U$ is unitary, it follows that
  its second column $u_2$ is a unit vector orthogonal to
  $e_1$. Therefore $u_2 = \lambda e_2$ where $\lambda$ is a unit of
  the Gaussian integers. Since the determinant of $W'$ is $i^m$ for
  some integer $m$, the determinant of $W'U$ is $i^{n+m}$, so that
  $\lambda=i^{n+m}$. Thus one of the following equalities must hold
  \begin{center}
    $W'U =I$, $ZW'U=I$, $SW'U=I$ or $ZSW'U =I$.
  \end{center}
  To prove the second claim, suppose that the least $\sqrt
  5$-denominator exponent of $U$ is $k$. Then $W$ can be efficiently
  computed because the algorithm described in the proofs of
  Lemma~\ref{lem-column-l} and Lemma~\ref{lem-column-k} requires
  $O(k)$ arithmetic operations.  Moreover, $W$ has $V$\!-count $k$ by
  Lemma~\ref{lem-column-k}, which is minimal since any Clifford+$V$
  circuit of $V$\!-count up to $k-1$ has least $\sqrt 5$-denominator
  exponent at most $k-1$.
\end{proof}

We conclude this section by noting that restricting $\ell$ to be equal
to 0 in (\ref{eq-matrix}) and the determinant of $U$ to be $\pm1$
yields a solution to the problem of exact synthesis in the Pauli+$V$
gate set.

\begin{proposition}
  \label{prop-exact-pauli}
  A unitary operator $U\in U(2)$ is exactly representable by a
  Pauli+$V$ circuit if and only if $U$ is of the form
  (\ref{eq-matrix}) with $\ell=0$ and $\det(U)=\pm 1$. Moreover, there
  exists an efficient algorithm that computes a Pauli+$V$ circuit for
  $U$ with $V$\!-count equal to the least $\sqrt 5$-denominator
  exponent of $U$, which is minimal.
\end{proposition}

\begin{proof}
  Analogous to the proof of Proposition~\ref{prop-exact-clifford},
  using the algorithm of Lemma~\ref{lem-column-k}.
\end{proof}

\section{Clifford+$V$ Approximate Synthesis of $z$-Rotations}
\label{sec-approx-synth}

In this section, we describe an algorithm to solve the problem of
approximate synthesis of $z$-rotations over the Clifford+$V$ gate set.

\begin{problem}
  \label{pb-approx-synth}
  Given an angle $\theta$ and a precision $\epsilon >0$, construct a
  Clifford+$V$ circuit $U$ whose $V$\!-count is as small as possible
  and such that $\norm{U-\Rz(\theta)}\leq \epsilon$.
\end{problem}

Our algorithm is adapted from the one developed in \cite{gridsynth}
for the Clifford+$T$ gate set. As in \cite{gridsynth}, we reduce
Problem~\ref{pb-approx-synth} to a pair of independent problems. From
Proposition~\ref{prop-exact-clifford}, we know that a unitary matrix
$U$ can be efficiently decomposed as a Clifford+$V$ circuit if and
only if
\begin{equation}
  U = \frac{1}{\rf k} \frac{1}{\rt \ell}
  \begin{pmatrix}
    \alpha    & \beta \\
    \gamma & \delta
  \end{pmatrix}, \quad \mbox{ with $k,\ell\in\N$,
    $\alpha,\beta,\gamma,\delta \in\Z[i]$, $0\leq \ell \leq 2$, and
    $\Det(U)=i^n$}.
\end{equation}
To solve Problem~\ref{pb-approx-synth}, we therefore need to find
$k,\ell\in\N$ and $\alpha,\beta,\gamma,\delta\in\Z[i]$ satisfying
these conditions and such that the resulting matrix $U$ approximates
$\Rz(\theta)$ up to $\epsilon$. The following lemma shows that we can
restrict our attention to matrices of determinant 1.

\begin{lemma}
  \label{lem-det-sol}
  If $\epsilon<|1-e^{i\pi/4}|$, then all solutions to
  Problem~\ref{pb-approx-synth} have the form
  \begin{equation}\label{eqn-u}
    U = \frac{1}{\rf k}\frac{1}{\rt \ell}
    \begin{pmatrix}
      \alpha & -\beta\da \\
      \beta & \alpha\da
    \end{pmatrix},
  \end{equation}
  with $k,\ell\in\N$, $\alpha,\beta\in\Z[i]$, and $0\leq \ell\leq
  2$. If $\epsilon\geq |1-e^{i\pi/4}|$, then there exists a solution
  of $V$\!-count 0 (i.e., a Clifford operator), and it is also of the
  form (\ref{eqn-u}).
\end{lemma}

\begin{proof}
  Every complex $2\by 2$ unitary operator $U$ can be written as
  \[
  U =
  \begin{pmatrix}
    a & -b\da e^{i\phi}\\
    b & a\da e^{i\phi}
  \end{pmatrix},
  \]
  for $a,b\in\C$ and $\phi\in[-\pi, \pi]$. This, together with the
  characterization of Clifford+$V$ operators given by
  Proposition~\ref{prop-exact-clifford}, implies that a complex $2\by
  2$ unitary operator $U$ can be exactly synthesized over the
  Clifford+$V$ basis if and only if
  \[
  U = \frac{1}{\rt k}\frac{1}{\rt \ell}
  \begin{pmatrix}
    \alpha    & -\beta\da i^n \\
    \beta & \alpha\da i^n
  \end{pmatrix},
  \]
  with $k,\ell,n\in\N$, $\alpha,\beta\in\Z[i]$, and $0\leq\ell\leq 2$.

  Now assume that $\epsilon<|1-e^{i\pi/4}|$ and
  $\norm{U-\Rz(\theta)}\leq \epsilon$. Let $e^{i\phi_1}$ and
  $e^{i\phi_2}$ be the eigenvalues of $U\Rz(\theta)\inv$, with
  $\phi_1,\phi_2\in[-\pi,\pi]$. Then
  \[
  |1-e^{i\pi/4}|>\epsilon\geq \norm{U-\Rz(\theta)} =
  \norm{I-U\Rz(\theta)\inv}=\max\s{|1-e^{i\phi_1}|,|1-e^{i\phi_2}|},
  \]
  so that $|1-e^{i\phi_j}|<|1-e^{i\pi/4}|$. Therefore
  $-\pi/4<\phi_j<\pi/4$, for $j\in\s{1,2}$, which implies that
  $-\pi/2<\phi_1+\phi_2<\pi/2$. Hence
  $|1-e^{i(\phi_1+\phi_2)}|<|1-e^{i\pi/2}|=\sqrt{2}$. But
  $e^{i(\phi_1+\phi_2)}=\det(U\Rz(\theta)\inv)=i^n$. Thus
  $|1-i^n|<\sqrt{2}$ which proves that $i^n=1$.
  
  For the last statement, note that if $\theta/2\in[-\pi/4,\pi/4]$,
  then $\norm{I-\Rz(\theta)}= |1-e^{i\theta/2}|\leq
  |1-e^{i\pi/4}|$. Similarly, if $\theta/2$ belongs to one of
  $[\pi/4,3\pi/4]$, $[3\pi/4,5\pi/4]$, or $[5\pi/4,7\pi/4]$, then one
  of $\norm{\omega^2-\Rz(\theta)}$, $\norm{-I-\Rz(\theta)}$, or
  $\norm{-\omega^2-\Rz(\theta)}$ is less than $|1-e^{i\pi/4}|$. In
  each case, $\Rz(\theta)$ is approximated to within $\epsilon$ by a
  Clifford operator.
\end{proof}

As a result of Lemma~\ref{lem-det-sol}, we know that to solve
Problem~\ref{pb-approx-synth}, it suffices to find $k,\ell\in\N$, with
$0\leq \ell \leq 2$, and $\alpha,\beta\in\Z[i]$ such that
$\alpha\da\alpha + \beta\da\beta = 5^k2^\ell$ and the resulting matrix
$U$ of the form (\ref{eqn-u}) approximates $\Rz(\theta)$ up to
$\epsilon$.  The key observation here is that, given $\epsilon$ and
$\theta$, we can express the requirement $\norm{U-\Rz(\theta)}\leq
\epsilon$ as a constraint on the top left entry
$\alpha/(\rf{k}\rt{\ell})$ of $U$. Indeed, let $z=e^{-i\theta/2}$,
$\alpha'=\alpha/(\rf{k}\rt{\ell})$, and
$\beta'=\beta/(\rf{k}\rt{\ell})$. Since ${\alpha'}\da
\alpha'+{\beta'}\da \beta'=1$ and $z\da z=1$, we have
\begin{align*}
  \norm{U - \Rz(\theta) }^2
  & = |\alpha'-z|^2 + |\beta'|^2 \\
  & = (\alpha'-z)\da(\alpha'-z) + {\beta'}\da\beta' \\
  & = {\alpha'}\da\alpha'+{\beta'}\da\beta' -
  z\da\alpha'-{\alpha'}\da z + z\da z  \\
  & = 2 - 2\Realpart(z\da \alpha').
\end{align*}
Thus $\norm{\Rz(\theta)-U}\leq \epsilon$ if and only if
$2-2\Realpart(z\da \alpha')\leq\epsilon^2$, or equivalently,
$\Realpart(z\da \alpha') \geq 1-\frac{\epsilon^2}{2}$. If we identify
the complex numbers $z=x+yi$ and $\alpha'=a+bi$ with 2-dimensional
real vectors $\vec z = (x,y)^T$ and $\vec \alpha' = (a,b)^T$, then
$\Realpart(z\da \alpha')$ is just their inner product $\vec z\cdot
\vec \alpha'$, and therefore $\norm{U - \Rz(\theta)}\leq \epsilon$ is
equivalent to
\begin{equation}
  \label{eqn-zalpha}
  \vec z\cdot \vec \alpha' \geq 1 - \frac{\epsilon^2}{2}.
\end{equation}
Moreover, ${\alpha'}\da \alpha'+{\beta'}\da \beta'=1$ implies that
${\alpha'}\da \alpha'= 1- {\beta'}\da \beta' \leq 1$ and therefore
that $\vec \alpha'$ is an element of the closed unit disk $\Disk$.
These two remarks jointly define a subset of the unit disk
\begin{equation}
  \label{eqn-Repsilon}
  \Repsilon= \s{\vec\alpha'\in\Disk 
    \such \vec z\cdot\vec\alpha'\geq 1-\frac{\epsilon^2}{2}}, 
\end{equation}
which we call the \emph{$\epsilon$-region} for $\theta$, such that if
$\alpha'\in\Repsilon$, then $\norm{U-\Rz(\theta)}\leq \epsilon$. In
the presence of $\alpha'=\alpha/(\rf{k}\rt{\ell})\in\Repsilon$, all
that remains is to find the other entry of $U$ by solving the
Diophantine equation
\[
\alpha\da\alpha + \beta\da\beta = 5^k2^\ell
\]
for some unknown $\beta\in\Z[i]$.

Now recall that we wish to solve Problem~\ref{pb-approx-synth}
optimally, so that we need to find an approximating matrix $U$ whose
$V$\!-count is as low as possible. We know from
Proposition~\ref{prop-exact-clifford} that the $V$\!-count of $U$ is
equal to its least $\sqrt 5$-denominator exponent. Therefore if we can
enumerate the points of $\Repsilon$ of the form
$\alpha/(\rf{k}\rt{\ell})$ for $\alpha\in\Z[i]$ and $0\leq \ell \leq
2$ in order of increasing $k$, then we can try to solve the
Diophantine equation for each such point. The first candidate for
which the Diophantine equation has a solution will then yield an
optimal solution to Problem~\ref{pb-approx-synth}.

Problem~\ref{pb-approx-synth} is therefore equivalent to the following
problem.

\begin{problem}
  \label{pb-approx-synth-bis}
  Given an angle $\theta$ and a precision $\epsilon>0$, find
  $k,\ell\in\N$ with $0\leq \ell \leq 2$ and $\alpha,\beta\in\Z[i]$
  such that:
  \begin{enumerate}[(i)]
  \item $\alpha/(\rf{k}\rt{\ell})\in\Repsilon$,
  \item $\alpha\da\alpha + \beta\da\beta = 5^k2^\ell$,
  \item and $k$ is as small as possible.
  \end{enumerate}
\end{problem}

In the above problem, the first two goals can be treated separately.

\begin{problem}[Scaled grid problem]
  \label{pb-scaled-grid}
  Given a bounded convex subset $A$ of $\R^2$ with non-empty interior,
  enumerate all points $\alpha/(\rf{k}\rt{\ell})\in A$, where
  $\alpha\in\Z[i]$, $k,\ell\in\N$, and $0\leq \ell \leq 2$, in order
  of increasing $(k,\ell)$.
\end{problem}

Each point $\alpha/(\rf{k}\rt{\ell})\in A$ is called a \emph{solution}
to the scaled grid problem for $A$ of denominator exponent $(k,\ell)$.

\begin{problem}[Diophantine equation]
  \label{pb-diophantine}
  Given $\alpha\in\Z[i]$ and $k,\ell\in\N$, find $\beta\in\Z[i]$ such
  that $\alpha\da\alpha +\beta\da\beta=5^k2^\ell$ if such a $\beta$
  exists.
\end{problem}

We now discuss methods to solve both of these problems. We provide an
algorithm for Problem~\ref{pb-approx-synth} and analyze its properties
in Section~\ref{ssect-algo-approx} and
Section~\ref{ssect-analysis-algo} respectively.

\subsection{Grid problems}
\label{ssect-grid-pbs}

In this subsection, we define an efficient algorithm to solve
Problem~\ref{pb-scaled-grid}. In what follows we refer to the set
$\Z^2\subseteq \R^2$ as the \emph{grid} and to elements of $\Z^2$ as
\emph{grid points}. The instances of the scaled grid problem where the
set $A$ is an upright rectangle, i.e., of the form $[x_1,x_2]\times
[y_1,y_2]$, are easy to solve. If $A$ is not an upright rectangle, the
problem can still be solved efficiently, provided that $A$ can be made
``upright enough''.

\begin{definition}[Uprightness]
  \label{def-upright}
  Let $A$ be a bounded convex subset of $\R^2$.  The bounding box of
  $A$, denoted $\BBox(A)$, is the smallest set of the form
  $[x_1,x_2]\by [y_1,y_2]$ that contains $A$. The \emph{uprightness of
    $A$}, denoted up(A), is defined to be the ratio of the area of A
  to the area of its bounding box:
  \[
  \up(A) =\frac{\area(A)}{\area(\BBox(A))}.
  \]
  We say that $A$ is $M$-upright if $\up(A) \geq M$.
\end{definition}

We will be especially interested in the case where the set $A$ is an
ellipse. Our interest in ellipses is motivated by the fact that a
bounded convex subset $A$ of the plane with non-empty interior can
always be enclosed in an ellipse whose area differs from that of $A$
by at most a constant factor. To increase the uprightness of a given
subset $A$ of the plane, we will then act on its ``enclosing ellipse''
using linear operators that map the grid to itself.

\begin{definition}[Ellipse]
  \label{def-ellipse}
  Let $D$ be a positive definite real $2\by 2$-matrix with non-zero
  determinant, and let $p\in\R^2$ be a point. The \emph{ellipse
    defined by $D$ and centered at $p$} is the set
  \[
  E = \s{u \in \R^2 \such (u-p)\da D(u -p)\leq 1}.
  \]
\end{definition}

\begin{proposition}
  \label{prop-enclosing-ellipse}
  Let $A$ be a bounded convex subset of $\R^2$ with non-empty
  interior.  Then there exists an ellipse $E$ such that $A\subseteq
  E$, and such that
  \[
  \area(E) \leq \frac{4\pi}{3\sqrt 3}\area(A).
  \]
\end{proposition}

\begin{proof}
  See theorems 5.17 and 5.18 of \cite{gridsynth}.
\end{proof}

The uprightness of an ellipse can be expressed in terms of the entries
of its defining matrix. Indeed, let $D$ be the positive definite
matrix defining some ellipse $E$ and assume that the entries of $D$
are as follows:
\[
D = \begin{pmatrix}
  a & b  \\
  b & d
\end{pmatrix}.
\]
We can compute the area of $E$ and the area of its bounding box using
$D$:
\[
\area(E)= \pi/\sqrt{\det(D)} \quad \mbox{and} \quad \area
(\BBox(E))=4\sqrt{ad} /\det(D).
\]
Thus by Definition~\ref{def-upright} we get:
\begin{equation}
  \label{eqn-upright-ellipse}
  \up(E) = \frac{\area(E)}
  {\area(\BBox(E))}
  = \frac{\pi}{4} 
  \sqrt{
    \frac{\det(D)}{ad} 
  }.
\end{equation}
The uprightness of $E$ is invariant under translation and scalar
multiplication.

\begin{definition}[Grid operator]
  \label{def-grid-operator}
  A \emph{grid operator} is an integer matrix, or equivalently, a
  linear operator, that maps $\Z^2$ to itself. A grid operator $\G$ is
  called \emph{special} if it has determinant $\pm 1$, in which case
  $G\inv$ is also a grid operator.
\end{definition}

\begin{remark}
  \label{rem-grid-of-ell}
  If $A$ is a subset of $\R^2$ and $\G$ is a grid operator, then
  $\G(A)$, the direct image of $A$, is defined as usual by
  $\G(A)=\s{\G(v)\such v\in A}$. If $\G$ is a grid operator and $E$ is
  an ellipse centered at the origin and defined by $D$, then $\G(E)$
  is an ellipse defined by $(\G\inv)\da D\G\inv$.
\end{remark}

\begin{proposition}
  \label{prop-to-upright}
  Let $E$ be an ellipse defined by $D$ and centered at $p$. There
  exists a grid operator $\G$ such that $\G(E)$ is $1/2$-upright.
  Moreover, if $E$ is $M$-upright, then $\G$ can be efficiently
  computed in $O(\log(1/M))$ arithmetic operations.
\end{proposition}

\begin{proof}
  If $E$ is an ellipse defined by a matrix $D$, we write $\sk(E)$ for
  the product of the anti-diagonal entries of $D$. Let $A$ and $B$ be
  the following special grid operators:
  \[
  A = \begin{pmatrix}
    1 & 1 \\
    0 & 1
  \end{pmatrix}
  , \quad B = \begin{pmatrix}
    1 & 0 \\
    1 & 1
  \end{pmatrix},
  \]
  and consider an arbitrary ellipse $E$. Since uprightness is
  invariant under translation and scaling, we may without loss of
  generality assume that $E$ is centered at the origin and that $D$
  has determinant 1. Suppose moreover that the entries of $D$ are as
  follows:
  \[
  \begin{pmatrix}
    a & b \\
    b & d
  \end{pmatrix}
  \]
  We first show that there exists a grid operator $\G$ such that
  $\sk(G(E))\leq 1$. Indeed, assume that $\sk(E) = b^2 \geq 1$. In
  case $a\leq d$, choose $n$ such that $|na+b|\leq a/2$. Then we have:
  \[
  {A^n}\da D A^n = \begin{pmatrix}
    \cdots & na+b \\
    na+b & \cdots
  \end{pmatrix}.
  \]
  Therefore, using Remark~\ref{rem-grid-of-ell} with $\G_1=(A^n)\inv$,
  we have:
  \[
  \sk(\G_1(E)) = (na+b)^2 \leq \frac{a^2}{4} \leq \frac{ad}{4} =
  \frac{1+b^2}{4} = \frac{1+\sk (E)}{4} \leq \frac{2~\sk (E)}{4} =
  \frac{1}{2}\sk(E).
  \]
  Similarly, in case $d<a$, then choose $n$ such that $|nd+b|\leq
  d/2$.  A similar calculation shows that in this case, with
  $\G_1=(B^n)\inv$, we get $\sk(\G_1(E)) \leq \frac{1}{2}\sk(E)$. In
  both cases, the skew of $E$ is reduced by a factor of 2 or
  more. Applying this process repeatedly yields a sequence of
  operators $\G_1,\ldots,\G_m$ and letting
  $\G=\G_m\cdot\ldots\cdot\G_1$ we find that $\sk(\G(E)) \leq 1$.
  
  Now let $D'$ be the matrix defining $\G(E)$, with entries as
  follows:
  \[
  D' =\begin{pmatrix}
    \alpha & \beta \\
    \beta & \delta
  \end{pmatrix}.
  \]
  Then $\sk(\G(E)) \leq 1$ implies that $\beta^2\leq 1$. Moreover,
  since $A$ and $B$ are special grid operators we have
  $\det(D')=\alpha\delta -\beta^2=1$. Using the expression
  (\ref{eqn-upright-ellipse}) for the uprightness of $\G(E)$ we get
  the desired result:
  \[
  \up(\G(E)) = \frac{\pi}{4} \sqrt{\frac{\Det(D')}{\alpha\delta}} =
  \frac{\pi}{4\sqrt{\alpha\delta}} = \frac{\pi}{4\sqrt{\beta^2+1}}
  \geq \frac{\pi}{4\sqrt{2}} \geq \frac{1}{2}.
  \]

  Finally, to bound the number of arithmetic operations, note that
  each application of $\G_j$ reduces the skew by at least a factor of
  2. Therefore, the number $n$ of grid operators required satisfies
  $n\leq \log_2(\sk(E))$. Now note that since $D$ has determinant 1,
  we have:
  \[
  M \leq \up(E) = \frac{\pi}{4}\frac{1}{\sqrt{ad}} =
  \frac{\pi}{4\sqrt{b^2+1}}.
  \]
  Therefore $\sk(E)=b^2\leq (\pi^2/16M^2)-1$, so that the computation
  of $\G$ requires $O(\log(1/M))$ arithmetic operations.
\end{proof}

We can now describe our algorithm to solve
Problem~\ref{pb-scaled-grid}. The algorithm inputs a bounded convex
set $A$ and we start by outlining the way in which the set $A$ is
given.

\begin{remark}
  \label{rem-ellipse}
  In the case of the present paper, a bounded convex set $A$ is
  \emph{given} if the following assumptions are satisfied.
  \begin{enumerate}[(i)]
  \item We are given an enclosing ellipse for $A$, whose area exceeds
    the area of $A$ by no more than a constant factor (such an ellipse
    exists by Proposition~\ref{prop-enclosing-ellipse}).
  \item We can efficiently decide, given $\alpha\in\Z[i]$ and
    $k,\ell\in\N$, whether or not $\alpha/\rf{k}\rt{\ell}$ belongs to
    $A$.
  \item We can efficiently compute the intersection of any straight
    line in $\Z[i, 1/\sqrt{5}, 1/\sqrt{2}]$ and $A$.
  \end{enumerate}
\end{remark}

\begin{proposition}
  \label{prop-algo-scaled-grid-fixed}
  There is an algorithm which, given a bounded convex subset $A$ of
  $\R^2$ with non-empty interior, enumerates all solutions of the grid
  problem for $A$ in order of increasing $(k,\ell)$. Moreover, if $A$
  is $M$-upright, then the algorithm requires $O(\log(1/M))$
  arithmetic operations overall, plus a constant number of arithmetic
  operations per solution produced.
\end{proposition}

\begin{proof}
  Given $A$ as in Remark~\ref{rem-ellipse}, with an enclosing ellipse
  $A'$ whose area only exceeds that of $A$ by a fixed constant factor
  $N$, use Proposition~\ref{prop-to-upright} to find a grid operator
  $\G$ such that $G(A')$ is $1/2$-upright. Then, enumerate the grid
  points of $\BBox(G(A'))$ in order of increasing $(k,\ell)$. This can
  be done efficiently since $\BBox(G(A'))$ is an upright
  rectangle. For each grid point $u$ found, check whether it belongs
  to $G(A)$. This is the case if and only if $G\inv(u)$ is a solution
  to the grid problem for $A$ with denominator exponent $(k,\ell)$.
\end{proof}

\subsection{Diophantine equations}
\label{ssect-dioph}

There is a well-known algorithm to solve Problem~\ref{pb-diophantine},
i.e., to solve the equation:
\begin{equation}\label{eqn-decomp1}
  \alpha\da\alpha + \beta\da\beta = 5^k2^\ell,
\end{equation}
for $\beta\in\Z[i]$, givenwhere $\alpha \in \Z[i]$ and
$k,\ell\in\N$. First note that if we write
$n=5^k2^\ell-\alpha\da\alpha$ and $\beta=b+ic$, where $n,b,c\in\Z$,
then Eq.~(\ref{eqn-decomp1}) is equivalent to
\begin{equation}\label{eqn-decomp}
  n = b^2+c^2.
\end{equation}
The solutions to Eq.~(\ref{eqn-decomp}) were characterized by Euler:

\begin{proposition}[Euler \cite{Euler}]
  \label{prop-algo-diophantine}
  Let $n$ be a positive integer with prime factorization
  $p_1^{k_1}\ldots p_m^{k_m}$, where $p_1,\ldots, p_m$ are distinct
  positive primes. Then $n$ can be written as the sum of two squares
  if and only if for all $i$ either $k_i$ is even or $p_i\equiv 1,2
  ~(\mymod 4)$.
\end{proposition}

\begin{proof}
  See Theorem~366 of \cite{Hardy-Wright}.
\end{proof}

Moreover, in case the equation $n=b^2+c^2$ has a solution, there is an
efficient probabilistic algorithm for finding $b$ and $c$, given a
prime factorization for $n$, see \cite{Rabin-Shallit}.

\subsection{The approximate synthesis algorithm}
\label{ssect-algo-approx}

We can now describe our algorithm to solve
Problem~\ref{pb-approx-synth}.

\begin{algorithm}
  \label{alg-main}
  Given $\theta$ and $\epsilon$, let $A=\Repsilon$ be the
  $\epsilon$-region as defined in Eq.~(\ref{eqn-Repsilon}).
  \begin{enumerate}[(i)]
  \item Use Proposition~\ref{prop-algo-scaled-grid-fixed} to enumerate
    the infinite sequence of solutions $\alpha/(\rf{k}\rt{\ell})$ to
    the scaled grid problem for $A$ in order of increasing least
    denominator exponent $(k,\ell)$.
  \item For each such solution $\alpha/(\rf{k}\rt{\ell})$ of least
    denominator exponent $(k,\ell)$:
    \begin{enumerate}
    \item[(a)] Let $n=5^k2^\ell - \alpha\da \alpha$.
    \item[(b)] Attempt to find a prime factorization of $n$. If $n\neq
      0$ but no prime factorization is found, skip step~(ii.c) and
      continue with the next $\alpha$.
    \item[(c)] Use the algorithm of Section~\ref{ssect-dioph} to solve
      the equation $\beta\da \beta = n$. If a solution $\beta$ exists,
      go to step~(iii); otherwise, continue with the next $\alpha$.
    \end{enumerate}
  \item Define $U$ as in Eq.~{\eqref{eqn-u}} and use the exact
    synthesis algorithm of Proposition~\ref{prop-exact-clifford} to
    find a Clifford+$V$ circuit for $U$. Output this circuit and stop.
  \end{enumerate}
\end{algorithm}

\begin{remark}
  \label{rem-pauli}
  By restricting $\ell$ to be equal to 0 throughout the algorithm and
  using Proposition~\ref{prop-exact-pauli} in step~(iii), we obtain a
  method for the approximate synthesis of $z$-rotations in the
  Pauli+$V$ basis.
\end{remark}

\subsection{Analysis of the algorithm}
\label{ssect-analysis-algo}

We now discuss the properties of Algorithm~\ref{alg-main}. The
restricted algorithm of Remark~\ref{rem-pauli} can be seen to enjoy
the same properties.

\subsubsection{Correctness}
\label{sssect-correctness}

\begin{proposition}
  \label{prop-correctness}
  If Algorithm~\ref{alg-main} terminates, then it yields a valid
  solution to the approximate synthesis problem, i.e., it yields a
  Clifford+$V$ circuit approximating $\Rz(\theta)$ up to $\epsilon$.
\end{proposition}

\begin{proof}
  By construction, following the reduction of
  Problem~\ref{pb-approx-synth} to Problem~\ref{pb-approx-synth-bis}.
\end{proof}

\subsubsection{Optimality in the presence of a factoring oracle}
\label{sssect-optimality}

\begin{proposition}
  \label{prop-optimality}
  In the presence of an oracle for integer factoring, the circuit
  returned by Algorithm~\ref{alg-main} has the smallest $V$\!-count of
  any single-qubit Clifford+$V$ circuit approximating $\Rz(\theta)$ up
  to $\epsilon$.
\end{proposition}

\begin{proof}
  By construction, step~(i) of the algorithm enumerates all solutions
  $\alpha$ to the scaled grid problem for $\Repsilon$ in order of
  increasing least $\sqrt 5$-denominator exponent $k$. Step~(ii.a)
  always succeeds and, in the presence of the factoring oracle, so
  does step~(ii.b). When step~(ii.c) succeeds, the algorithm has found
  a solution of Problem~\ref{pb-approx-synth-bis} for a minimal $k$.
\end{proof}

\subsubsection{Near-optimality in the absence of a factoring oracle}
\label{sssect-near-optimality}

The proof that our algorithm is nearly optimal in the absence of a
factoring oracle relies on the following number-theoretic
hypothesis. We do not have a proof of this hypothesis, but it appears
to be valid in practice.

\begin{hypothesis}
  \label{hyp-numbers}
  For each number $n$ produced in step~(ii.a) of
  Algorithm~\ref{alg-main}, write $n=2^jm$, where $m$ is odd. Then $m$
  is asymptotically as likely to be a prime congruent to 1 modulo 4 as
  a randomly chosen odd number of comparable size. Moreover, each $m$
  can be modelled as an independent random variable.
\end{hypothesis}

\begin{lemma}
  \label{lem-evolution-grid}
  Let $A$ be a bounded convex subset of $\R^2$, $k\geq 0$, and assume
  that the scaled grid problem for $A$ has at least two distinct
  solutions with $\sqrt 5$-denominator exponent $k$. Then for all
  $j\geq 0$, the scaled grid problem for $A$ has at least $5^j+1$
  solutions with $\sqrt 5$-denominator exponent $k+2j$.
\end{lemma}

\begin{proof}
  Let $\alpha\neq \beta$ be solutions of the scaled grid problem for
  $A$ with $\sqrt 5$-denominator exponent $k$. For each
  $\ell=0,1,\ldots,5^j$, let $\phi=\frac{\ell}{5^j}$, and consider
  $\alpha_j = \phi \alpha + (1-\phi) \beta$. Then $\alpha_j$ has
  $\sqrt 5$-denominator exponent $k+2j$. Also, $\alpha_j$ is a convex
  combination of $\alpha$ and $\beta$. Since $A$ is convex, it follows
  that $\alpha_j$ is a solution of the scaled grid problem for $A$,
  yielding $5^j+1$ distinct solutions with $\sqrt 5$-denominator
  exponent $k+2j$.
\end{proof}

\begin{lemma}
  \label{lem-sum}
  Fix an arbitrary constant $b>0$. Then for $a\geq 1$,
  \[
  \sum_{x=1}^{\infty} \bigparen{1-\frac{1}{a+b\ln x}}^x = O(a).
  \]
\end{lemma}

\begin{proof}
  The lemma is proved in Appendix E of \cite{gridsynth}.
\end{proof}

\begin{definition}
  Let $U'$ and $U''$ be the following two solutions of the approximate
  synthesis problem
  \begin{equation}\label{eqn-u1-u2}
    U' = \zmatrix{cc}{\alpha' & -\beta'^{\dagger} \\ 
      \beta' & \alpha'^{\dagger}}
    \quad 
    \mbox{and}
    \quad
    U'' = \zmatrix{cc}{\alpha'' & -\beta''^{\dagger} \\ 
      \beta'' & \alpha''^{\dagger}}.
  \end{equation}
  $U'$ and $U''$ are said to be {\em equivalent solutions} if
  $\alpha'=\alpha''$.
\end{definition}

\begin{proposition}
  \label{prop-near-optimality}
  Let $k$ be the $V$\!-count of the solution of the approximate
  synthesis problem found by Algorithm~\ref{alg-main} in the absence
  of a factoring oracle. Then
  \begin{enumerate}[(i)]
  \item The approximate synthesis problem has at most
    $O(\log(1/\epsilon))$ non-equivalent solutions with $V$\!-count
    less than $k$.
  \item The expected value of $k$ is $k'''+O(\log(\log(1/\epsilon)))$,
    where $k', k'',$ and $k'''$ are the $V$\!-counts of the optimal,
    second-to-optimal, and third-to-optimal solutions of the
    approximate synthesis problem (up to equivalence).
  \end{enumerate}
\end{proposition}

\begin{proof}
  If $\epsilon\geq |1-e^{i\pi/4}|$, then by Lemma~\ref{lem-det-sol}
  there is a solution of $V$\!-count 0 and the algorithm easily finds
  it.  In this case there is nothing to show, so assume without loss
  of generality that $\epsilon<|1-e^{i\pi/4}|$. Then by
  Lemma~\ref{lem-det-sol}, all solutions are of the form
  (\ref{eqn-u}).
  \begin{enumerate}[(i)]
  \item Consider the list $\alpha_1,\alpha_2,\ldots$ of candidates
    generated in step~(i) of the algorithm. Let $k_1,k_2,\ldots$ be
    their least $\sqrt 5$-denominator exponent and let
    $n_1,n_2,\ldots$ be the corresponding integers calculated in
    step~(ii.a). Note that $n_j\leq 4\cdot 5^{k_j}$ for all $j$. Write
    $n_j=2^{z_j}m_j$ where $m_j$ is odd. By
    Hypothesis~\ref{hyp-numbers}, the probability that $m_j$ is a
    prime congruent to 1 modulo 4 is asymptotically no smaller than
    that of a randomly chosen odd integer less than $4\cdot 5^{k_j}$,
    which, by the well-known prime number theorem, is
    \begin{equation}\label{eqn-pj}
      p_j:=\frac{1}{\ln (4\cdot 5^{k_j})} 
      = \frac{1}{k_j\ln 5 + \ln 4}.
    \end{equation}
    By the pigeon-hole principle, two of $k_1,k_2,$ and $k_3$ must be
    congruent modulo 2. Assume without loss of generality that
    $k_2\equiv k_3 ~(\mymod 2)$. Then $\alpha_2$ and $\alpha_3$ are
    two distinct solutions to the scaled grid problem for $\Repsilon$
    with (not necessarily least) denominator exponent $k_3$. It
    follows by Lemma~\ref{lem-evolution-grid} that there are at least
    $5^r+1$ distinct candidates of denominator exponent $k_3+2r$, for
    all $r\geq 0$. In other words, for all $j$, if $j\leq 5^r+1$, we
    have $k_j\leq k_3+2r$.  In particular, this holds for
    $r=\floor{1+\log_5 j}$, and therefore,
    \begin{equation}\label{eqn-kj}
      k_j\leq k_3+2(1+\log_5 j).
    \end{equation}
    Combining {\eqref{eqn-kj}} with {\eqref{eqn-pj}}, we have
    \begin{equation}\label{eqn-pj2}
      p_j \geq \frac{1}{(k_3+2(1+\log_5 j))\ln 5 + \ln 4}
      = \frac{1}{(k_3+2)\ln 5 + 2\ln j + \ln 4}
    \end{equation}
    Let $j_0$ be the smallest index such that $m_{j_0}$ is a prime
    congruent to 1 modulo 4.  By Hypothesis~\ref{hyp-numbers}, we can
    treat each $m_j$ as an independent random variable. Therefore,
    \begin{eqnarray}
      P(j_0 > j) &=& P(\mbox{$n_1,\ldots,n_j$ are not prime})
      \nonumber\\
      &\leq & (1-p_1)(1-p_2)\cdots(1-p_j) \nonumber\\
      &\leq & (1-p_j)^j \nonumber\\
      &\leq & 
      \bigparen{1-\frac{1}{(k_3+2)\ln 5 + 2\ln j + \ln 4}}^j.         
      \nonumber
    \end{eqnarray}
    The expected value of $j_0$ is
    \begin{equation}\label{eqn-ej0}
      E(j_0) ~=~ \sum_{j=0}^{\infty} P(j_0 > j)
      ~\leq~ 1+\sum_{j=1}^{\infty} 
      \bigparen{1-\frac{1}{(k_3+2)\ln 5 + 2\ln j + \ln 4}}^j
      ~=~ O(k_3),
    \end{equation}
    where we have used Lemma~\ref{lem-sum} to estimate the sum.
          
    Next, we will estimate $k_3$. First note that if the $\epsilon$
    region contains a circle of radius greater than $1/\rf{k}$, then
    it contains at least 3 solutions to the scaled grid problem for
    $\Repsilon$ with denominator exponent $k$. The width of the
    $\epsilon$-region $\Repsilon$ is $\epsilon^2/2$ at the widest
    point, and we can inscribe a disk of radius $r={\epsilon^2}/{4}$
    in it. Hence the scaled grid problem for $\Repsilon$, as in
    step~(i) of the algorithm, has at least three solutions with
    denominator exponent $k$, provided that
    \[
    r = \frac{\epsilon^2}{4} \geq \frac{1}{\rf k},
    \]
    or equivalently, provided that
    \[
    k\geq 2 \log_5(2) + 2\log_5 ({1}/{\epsilon}).
    \]
    It follows that
    \begin{equation}\label{eqn-k2}
      k_3=O(\log({1}/{\epsilon})),
    \end{equation}
    and therefore, using {\eqref{eqn-ej0}}, also
    \begin{equation}\label{eqn-ej0-eps}
      E(j_0)=O(\log({1}/{\epsilon})).
    \end{equation}
          
    To finish the proof of part~(i), recall that $j_0$ was defined to
    be the smallest index such that $m_{j_0}$ is a prime congruent to
    1 modulo 4. The primality of $m_{j_0}$ ensures that step~(ii.b) of
    the algorithm succeeds for the candidate
    $\alpha_{j_0}$. Furthermore, because $m_{j_0}\equiv 1\mmod{4}$,
    the equation $\beta\da \beta=n$ has a solution by
    Proposition~\ref{prop-algo-diophantine}. Hence the remaining steps
    of the algorithm also succeed for $\alpha_{j_0}$.

    Now let $s$ be the number of non-equivalent solutions of the
    approximate synthesis problem of $V$\!-count strictly less than
    $k$. As noted above, any such solution $U$ is of the form
    {\eqref{eqn-u}}. Then the least denominator exponent of $\alpha$
    is strictly smaller than $k_{j_0}$, so that $\alpha=\alpha_j$ for
    some $j<j_0$. In this way, each of the $s$ non-equivalent
    solutions is mapped to a different index $j<j_0$. It follows that
    $s<j_0$, and hence that $E(s)\leq E(j_0)=O(\log({1}/{\epsilon}))$,
    as was to be shown.
  \item Let $U'$ be an optimal solution of the approximate synthesis
    problem, let $U''$ be optimal among the solutions that are not
    equivalent to $U'$ and let $U'''$ be optimal among the solutions
    that are not equivalent to either $U'$ or $U''$. Assume that $U',
    U'',$ and $U'''$ are written as in (\ref{eqn-u1-u2}) with top-left
    entry $\alpha',\alpha'',$ and $\alpha'''$ respectively. Now let
    $k'$, $k''$, and $k'''$ be the least denominator exponents of
    $\alpha'$, $\alpha''$, and $\alpha'''$, respectively. Let $k_3$
    and $j_0$ be as in the proof of part~(i). Note that, by
    definition, $k_3\leq k'''$.  Let $k$ be the least denominator
    exponent of the solution of the approximate synthesis problem
    found by the algorithm. Then $k\leq k_{j_0}$. Using
    {\eqref{eqn-kj}}, we have
    \[
    k \leq k_{j_0} \leq k_3 + 2(1+\log_5 j_0) \leq k''' + 2(1+\log_5
    j_0).
    \]
    This calculation applies to any one run of the algorithm.  Taking
    expected values over many randomized runs, we therefore have
    \begin{equation}\label{eqn-em}
      E(k) \leq k''' + 2 + 2 E(\log_5 j_0) 
      \leq k''' + 2 + 2\log_5 E(j_0).
    \end{equation}
    Note that we have used the law $E(\log j_0)\leq \log(E(j_0))$,
    which holds because $\log$ is a concave function. Combining
    {\eqref{eqn-em}} with {\eqref{eqn-ej0-eps}}, we therefore have the
    desired result:
    \begin{equation}
      E(k) = k''' + O(\log(\log(1/\epsilon))). \nonumber
    \end{equation}\qedhere
  \end{enumerate}
\end{proof}

\subsubsection{Time complexity}
\label{sssect-complexity}

\begin{proposition}
  Algorithm~\ref{alg-main} runs in expected time
  $O(\polylog(1/\epsilon))$. This is true whether or not a
  factorization oracle is used.
\end{proposition}

\begin{proof}
  This proposition is proved like the corresponding one in
  \cite{gridsynth}.
\end{proof}

\section{Conclusion}

We have introduced an algorithm for the approximate synthesis of
$z$-rotations into Clifford+$V$ circuits. Our algorithm is optimal if
an oracle for the factorization of integers is available. In the
absence of such an oracle, our algorithm is still nearly optimal,
yielding circuits of $V$\!-count $m + O(\log(\log(1/\epsilon)))$,
where $m$ is the $V$\!-count of the third-to-optimal solution. We have
also described an algorithm for the approximate synthesis of
$z$-rotations into Pauli+$V$ circuits. To the author's knowledge,
these algorithms are the first optimal synthesis algorithms for
extensions of the $V$\!-gates.

\section*{Acknowledgements}

The author would like to thank Peter Selinger and Kira Scheibelhut for
their helpful comments.

\bibliographystyle{abbrv} 
\bibliography{vsynth}

\end{document}